\documentclass[runningheads]{llncs}
\usepackage{graphicx}
\newcommand{\be}[1]{{\color{blue} #1}}

\usepackage{wrapfig}

\usepackage[linesnumbered,lined,boxed,commentsnumbered, ruled]{algorithm2e}
\usepackage{multirow}
\usepackage{tabularx}
\usepackage{comment}
\usepackage{amsfonts}
\usepackage{amsmath}

\usepackage{amsthm}

\usepackage[hyphens]{url}
\usepackage{hyperref}
\hypersetup{breaklinks=true}

\begin{document}
\title{Canadian Traveller Problem with Predictions}%
\author{Evripidis Bampis\inst{1} \and
Bruno Escoffier\inst{1,2} \and
Michalis Xefteris\inst{1}}
\authorrunning{E. Bampis et al.}
\institute{Sorbonne Universit\'e, CNRS, LIP6, F-75005 Paris, France \and
Institut Universitaire de France, Paris, France}
\maketitle              % typeset the header of the contribution
\begin{abstract}
In this work, we consider the $k$-Canadian Traveller Problem ($k$-CTP) under the learning-augmented framework proposed by Lykouris \& Vassilvitskii~\cite{Lykouris}.  $k$-CTP is a generalization of the shortest path problem, and involves a traveller who knows the entire graph in advance and wishes to find the shortest route from a source vertex $s$ to a destination vertex $t$, but discovers online that some edges (up to $k$) are blocked once reaching them. A potentially imperfect predictor gives us the number and the locations of the blocked edges.

We present a deterministic and a randomized online algorithm for the learning-augmented $k$-CTP that achieve a tradeoff between consistency (quality of the solution when the prediction is correct) and robustness (quality of the solution when there are errors in the prediction). Moreover, we prove a matching lower bound for the deterministic case establishing that the tradeoff between consistency and robustness is optimal, and show a lower bound for the randomized algorithm.
Finally, we prove several deterministic and randomized lower bounds on the competitive ratio of $k$-CTP depending on the prediction error, and complement them, in most cases, with matching upper bounds.

\end{abstract}

\section{Introduction}
Motivated by various applications including online route planning in road networks, or message routing in communication networks, the Canadian Traveller problem (CTP), introduced in 1991 by Papadimitriou and Yannakakis \cite{PY}, is a generalization of one of the most prominent problems in Computer Science, the Shortest Path Problem \cite{Lawler,Papadimitriou82}. In CTP the underlying graph is given in advance, but it is unreliable, i.e. some edges may become unavailable (e.g. because of snowfall, or link failure) in an online manner. The blockage of an edge becomes known to the algorithm only when it arrives at one of its extremities. The objective  is to devise an efficient adaptive strategy minimizing the ratio between the length of the path found and the optimum (where the blocked edges are removed).
Papadimitriou and Yannakakis \cite{PY} proved that the problem of devising an algorithm that guarantees a given competitive ratio is PSPACE-complete if the number of blocked edges is not fixed. %Since then,  CTP has been studied in the case where the number of blocked edges is bounded by $k$. This problem is known as $k$-CTP \cite{BNS}. 
Given the intractability of CTP, Bar-Noy and Schieber   \cite{BNS} focused on $k$-CTP, a special case of CTP where the number of blocked edges is bounded by $k$. Here, we consider $k$-CTP in the framework of {\em learning-augmented} online algorithms \cite{Lykouris,MV}. It is natural to consider that in applications, like route planning, or message routing in communication networks, predictions may be provided  on the input data.
%Our aim is to study whether predictions may  be used in order to improve the quality of  online algorithms for $k$-CTP. 
Our aim is to study the impact of the quality of such predictions on the performance  of online algorithms for $k$-CTP.
%\subsection{\textit{k}-Canadian Traveller Problem}

Formally, in $k$-CTP we consider a connected undirected graph $\mathcal{G}=(V,E)$ with a source node $s$, a destination node $t$ and a non-negative cost function $c: E \rightarrow \mathbb{R^+}$ representing the cost to traverse each edge. An agent seeks to find a shortest path from $s$ to $t$. However, one or more edges (up to $k$) might be blocked, and thus cannot be traversed. An agent only learns that an edge is blocked when reaching one of its endpoints.

In classical competitive analysis,  a deterministic online algorithm $ALG$ for $k$-CTP is $c$-competitive if the total length $ALG(\sigma)$ traversed by $ALG$ for input $\sigma$ is at most $c \cdot OPT(\sigma)$, where $OPT(\sigma)$ is the length of a shortest $s-t$ path in $G$ without the blocked edges~\cite{sleator}. A randomized algorithm is $c$-competitive against an oblivious adversary if the expected cost $\mathbb{E}[ALG(\sigma)]$ is at most $c \cdot OPT(\sigma)$~\cite{Borodin}. 

Bar-Noy and Schieber~\cite{BNS} considered $k$-CTP and they proposed a polynomial time algorithm that minimizes the maximum travel length.
Westphal in~\cite{WESTPHAL} gave a simple online deterministic algorithm for $k$-CTP which is $(2k+1)$-competitive. He also proved that no deterministic online algorithm with a better competitive ratio exists. Furthermore, he showed a lower bound for any randomized algorithm of $k+1$, even if all $s-t$ paths are node disjoint. Xu et al.~\cite{Xu} proposed a deterministic algorithm that is also $(2k+1)$-competitive for $k$-CTP. They also proved that  a natural greedy strategy based on the available blockage information is exponential in $k$. %Bender and Westphal~\cite{bender} devised a $(k+1)$-competitive randomized online algorithm for $k$-CTP on graphs where all $s-t$ paths are node-disjoint.
A $(k+1)$-competitive randomized online algorithm for $k$-CTP is known on graphs where all $s-t$ paths are node-disjoint~\cite{bender,Shiri}.
Demaine et al.~\cite{Demaine} proposed a polynomial time randomized algorithm that improves the deterministic lower bound of $2k+1$ by an $o(1)$ factor for arbitrary graphs. They also showed that the competitive ratio is even better if the randomized algorithm runs in pseudo-polynomial time. More recently, Berg\'e et al.~\cite{berge} proved that the competitive ratio of any randomized memoryless strategy (agent's strategy does not depend on his/her anterior moves) cannot be better than $2k + O\left( 1\right)$. Several other variants of the problem have been studied in the recent years~\cite{Huang}, \cite{KN}.

Given the widespread of Machine Learning technology, in the last years, predictions from ML are used in order to improve the worst case analysis of online algorithms~\cite{Devenur,Vee,Cole,Medina}. The formal framework for these learning-augmented algorithms has been presented by Lykouris and Vassilvitskii in their seminal paper~\cite{Lykouris}, where they studied the caching problem. In this framework, no assumption is made concerning the quality of the predictor and the challenge is to design a learning-augmented online algorithm that finds a good tradeoff between the two extreme alternatives, i.e. following blindly the predictions, or simply ignore them. 
Ideally, the objective is to produce algorithms using predictions that are {\em consistent}, i.e. whose performance is close to the best offline algorithm when the prediction is accurate,  and {\em robust}, i.e. whose performance is close to the online algorithm without predictions when the prediction is bad.

% This gives rise to the framework of learning-augmented online algorithms \cite{Lykouris}. In this framework, no assumption is made concerning the quality of the predictor and the challenge is to design a learning-augmented online algorithm that finds a good trade-off between the two extreme alternatives, i.e. following blindly the predictions, or simply ignore them. 
% Ideally, the objective is to produce algorithms using predictions that are {\em consistent}, i.e. whose performance is close to the best offline algorithm when the prediction is accurate,  and {\em robust}, i.e. whose performance is close to the online algorithm without predictions when the prediction is bad.

% The idea of using predictions in order to improve algorithm's performance has been studied in a series of papers~\cite{Devenur,Vee,Cole,Medina}, but the formal model of learning-augmented algorithms has been presented by Lykouris and Vassilvitskii in their seminal paper~\cite{Lykouris}, where they studied the caching problem.

Antoniadis et al. \cite{Antoniadis20}, Rohatgi~\cite{Rohatgi} and Wei~\cite{wei2} subsequently gave simpler and improved algorithms for the caching problem. Kumar et al. in~\cite{Kumar} applied the learning-augmented setting to ski rental and online scheduling. For the same problems, Wei and Zhang, \cite{Wei}, provided a set of non-trivial lower bounds for competitive analysis of learning-augmented online algorithms. Many other papers have been published in this direction for ski rental \cite{Gollapudi,WangLW20,Banerjee20,Bamas2}, scheduling  \cite{AzarLT21,Bamas1,Lattanzi20,Mitzenmacher20,Im}, the online $k$-server problem~\cite{Lindermayr}, $k$-means clustering~\cite{k-means} and others~\cite{AntoniadisGKK20,LuRSZ21,LavastidaM0X21}.

\subsection{Our Contribution}
% Change order
In this work, we study the $k$-Canadian Traveller Problem through the lens of online algorithms with predictions. We present both deterministic and randomized upper and lower bounds for the problem. Following previous works we focus on path-disjoint graphs for the randomized case\footnote[1]{As mentioned in earlier, while a $(2k+1)$-competitive (matching the lower bound) deterministic algorithm is known for general graph, a $(k+1)$-competitive randomized algorithm (matching the lower bound) is only known for path-disjoint graphs.}. We use a simple model where we are given predictions on the locations of the blocked edges. For example, consider a situation wherein you need to follow the shortest route to a destination. You usually open the Maps app on your phone to find you the best route. Maps app does that using predictions about the weather condition, the traffic jam etc. These predictions capture additional side information about the route we should follow. In our model, the error of the prediction is just the total number of false predictions we get. The parameter $k$ upper bounds the number of real blocked edges (denoted by $\kappa$, usually unknown when the algorithm starts) and the number of predicted blocked edges (denoted by $k_p$) in the graph, meaning that we want to design algorithms that hedge against all situations where up to $k$ edges can be blocked (and up to $k$ predicted to be blocked).

In Section~\ref{sec:tradeoffs}, we give the main results of this paper which are algorithms with predictions for $k$-CTP (deterministic and randomized ones) that are as consistent and robust as possible. 
More precisely, we say that an algorithm is $(a, b]$-competitive ($[a,b]$-competitive), when it achieves a competitive ratio smaller than (no more than) $a$ when the prediction is correct and no more than $b$ otherwise. 
%this is an adjustment of the classical  where ratios are generally function of the error. smoothly increase with the error of the prediction to the difficulty of $k$-CTP. %The adjustment is explained by the results we present in the next section ().  
 Our aim is to answer the following question: if we want an algorithm which is $(1+\epsilon)$-competitive if the prediction is correct (consistency $1+\epsilon$), what is the best competitive ratio we can get when the prediction is not correct (robustness)?
The parameter $\epsilon > 0$ is user defined, possibly adjusted depending on her/his level of trust in the predictions.
%Given a parameter $\epsilon > 0$, we want to determine the best robustness that can be achieved given that we want a consistency of at most $1+\epsilon$.  
The results are presented in Table~\ref{table:1}. We give a deterministic $\big(1+\epsilon,2k-1+\frac{4k}{\epsilon}\big]$-lower bound and a matching upper bound. For the randomized case, we give a randomized $\big[1+\epsilon,k+\frac{k}{\epsilon}\big]$-lower bound and an algorithm that achieves a tradeoff of $\big[1+\epsilon,k+\frac{4k}{\epsilon}\big]$ on path-disjoint graphs, when $k$ is considered as known. We note that the above lower bounds are also valid when the parameter is $\kappa$ (the real number of blocked edges). In most real world problems such as the ones described earlier, the number of blocked edges is usually small and we can get interesting tradeoffs between consistency and robustness.

\begin{table}[h!]
\caption{Our bounds on the tradeoffs between consistency and robustness for our learning-augmented model ($0 < \epsilon \leq 2k$ for the deterministic case, and $0 < \epsilon \leq k$ for the randomized one).}
\centering
\setlength{\extrarowheight}{0.1cm}

\begin{tabular}{c|c|c}
\hline
\multirow{3}{6em}{\begin{tabular}[c]{@{}c@{}}Deterministic\\ algorithms\end{tabular}} & Lower bound & \begin{tabular}[c]{@{}c@{}}\big($1+\epsilon$, $2k-1+\frac{4k}{\epsilon}$\big]\\ Theorem~\ref{theor_det_e}\end{tabular} \\ \cline{2-3} 
                                                                                    & Upper bound & \begin{tabular}[c]{@{}c@{}}\big($1+\epsilon$, $2k-1+\frac{4k}{\epsilon}$\big]\\ Theorem~\ref{alg_det_e}\end{tabular} \\ \hline \hline
\multirow{3}{6em}{\begin{tabular}[c]{@{}c@{}}Randomized\\ algorithms\end{tabular}}    & Lower bound & \begin{tabular}[c]{@{}c@{}}\big[$1+\epsilon$, $k+\frac{k}{\epsilon}$\big]\\ Theorem~\ref{theor_rand_e}\end{tabular} \\ \cline{2-3} 
                                                                                    & Upper bound & \begin{tabular}[c]{@{}c@{}}\big[$1+\epsilon$, $k+\frac{4k}{\epsilon}$\big]\\  Theorem~\ref{alg_rand_e}\end{tabular} \\ \hline
\end{tabular}
\label{table:1}
\end{table}

In Section~\ref{sec:robustness}, we explore the competitive ratios of $k$-CTP that can be achieved depending on the error of the predictor.
Besides consistency and robustness, most works in this area classically study smooth error dependencies for the competitive ratio \cite{Bamas2,Rohatgi}. In this paper this is not the case, since the error is highly non-continuous. 
Our analysis contains both deterministic and randomized lower bounds complemented with matching upper bounds in almost all cases. These results are presented in Table~\ref{table:2} and justify the model of consistency-robustness tradeoff we chose in the previous section. All lower bounds are also valid if the parameter is $\kappa$. Note that for all upper bounds (except for $c^*=2k+1$) the parameter $k$ is considered to be known in advance. For the randomized case, the upper bound is given for path-disjoint graphs.

\begin{table}[h!]
\caption{Our bounds on the competitive ratio $c^*$ with respect to $error$, $k$. The upper bounds $2k+1$ and $k+1$ are also valid for $error \leq t$, for any $t\geq 2$.
}
\centering
\setlength{\extrarowheight}{0.13cm}

\begin{tabular}{cc|ccc}
\multicolumn{2}{c|}{}                                                                                                  & \multicolumn{1}{c|}{$k=1$}                                                    & \multicolumn{1}{c|}{$k=2$}                                                       & $k \geq 3$                                      \\ \hline
\multicolumn{1}{c|}{\multirow{3}{6em}{\begin{tabular}[c]{@{}c@{}}Deterministic\\ algorithms\end{tabular}}} & $error \leq 1$ & \multicolumn{1}{c|}{\begin{tabular}[c]{@{}c@{}}$c^*=3$\\ Theorem~\ref{theor_k_1}\end{tabular}} & \multicolumn{1}{c|}{\begin{tabular}[c]{@{}c@{}}$c^*=\frac{3+\sqrt{17}}{2}\simeq 3.56$\\ Theorems~\ref{theoremLBk=2}~and~\ref{theoremUBk=2}\end{tabular}} & \begin{tabular}[c]{@{}c@{}}$c^*=2k-1$\\ Theorems~\ref{theorem2k-1}~and~\ref{theorem2}\end{tabular} \\ \cline{2-5} 
\multicolumn{1}{c|}{}                                                                              & $error \leq 2$ & \multicolumn{3}{c}{\begin{tabular}[c]{@{}c@{}}$c^*=2k+1$\\ Theorem~\ref{theorem1}\end{tabular}}                                                                                                                                             \\ \hline \hline
\multicolumn{1}{c|}{\multirow{3}{6em}{\begin{tabular}[c]{@{}c@{}}Randomized\\ algorithms\end{tabular}}}    & $error \leq 1$ & \multicolumn{3}{c}{\begin{tabular}[c]{@{}c@{}}$c^*\geq k$\\ Theorem~\ref{th:randk}\end{tabular}}                                                                                                                                                \\ \cline{2-5} 
\multicolumn{1}{c|}{}                                                                              & $error \leq 2$ & \multicolumn{3}{c}{\begin{tabular}[c]{@{}c@{}}$c^*=k+1$\\ Theorem~\ref{theor_rand_k1}\end{tabular}}                                                                                                                                              \\ \hline
\end{tabular}
\label{table:2}
\end{table}

\section{Preliminaries}
We introduce two algorithms of the literature that are useful for our work and some notation we use in the rest of the technical sections.

\subsection{Deterministic and randomized algorithms}
As mentioned in the introduction, Westphal in~\cite{WESTPHAL} gave an optimal deterministic algorithm \textsc{Backtrack} for $k$-CTP, which is $(2k+1)$-competitive (note that the algorithm does not need to know $k$).

\textsc{Backtrack}: An agent begins at source $s$ and follows the cheapest $s-t$ path on the graph. When the agent learns about a blocked edge on the path to $t$, he/she returns to $s$ and takes the cheapest $s-t$ path without the blocked edge discovered. The agent repeats this strategy until he/she arrives at $t$. Observe that he/she backtracks at most $k$ times, since there are no more than $k$ edges blocked, and thus \textsc{Backtrack} is $(2k+1)$-competitive.

%\subsection{Randomized algorithm}

\begin{figure}%{r}{0.3\textwidth}
\centering
    \includegraphics[scale=0.35]{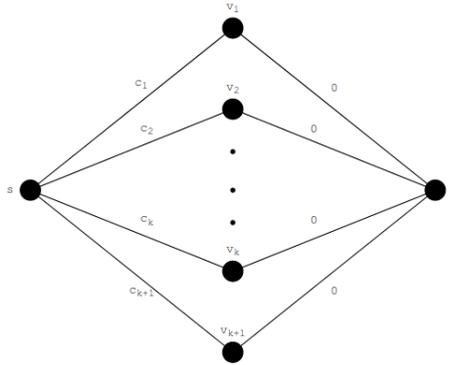}
    \caption[]{The graph $\mathcal{G}^*$ for the proof of lower bounds.}
    \label{fig:graph}
\end{figure}

%\begin{wrapfigure}{r}{0.3\textwidth}
%    \includegraphics[scale=0.35]{graph_disjoint}
%    \caption[]{The graph $\mathcal{G}^*$ for the proof of lower bounds.}
%    \label{fig:graph}
%\end{wrapfigure}

Concerning randomized algorithm, the proof of the lower bound of $(k+1)$~\cite{WESTPHAL}, holding even if all $s-t$ paths are node-disjoint (besides $s$ and $t$), uses the graph in Figure~\ref{fig:graph} with $c_1=c_2=\dots=c_{k+1}$ (note that the lower bound holds when restricting to positive length  by replacing $0$ with $\epsilon>0$). The matching upper bound of $(k+1)$~\cite{bender,Shiri}, which is known to hold only when the paths are node-disjoint, is based on a randomized algorithm that we will call \textsc{RandBacktrack} in the sequel. The very general idea of this algorithm, which can be seen as a randomized version of \textsc{Backtrack}, is the following. 

\textsc{RandBacktrack}: Consider the $k+1$ shortest $s-t$ paths in the graph. The algorithm defines an appropriate probability distribution and then chooses a path according to this distribution that the agent tries to traverse. If this path is feasible, the algorithm terminates. If it is blocked, the agent returns to $s$ and repeats the procedure for a smaller set of paths.

%it has been shown in~\cite{WESTPHAL} that no randomized online algorithm can be better than $(k+1)$-competitive, even if all $s-t$ paths are node-disjoint (besides $s$ and $t$). The graph for the proof of the lower bound is the same as in Figure~\ref{fig:graph} with $c_1=c_2=\dots=c_{k+1}$ (note that the lower bound holds when restricting to positive length  by replacing $0$ with $\epsilon>0$). 

%Bender and Westphal~\cite{bender} constructed a randomized online algorithm for the $k$-CTP on graphs where all $s-t$ paths are node-disjoint, where $k$ is considered as known and showed that it achieves a competitive ratio of $k+1$ against an oblivious adversary. Although Shiri et al.~\cite{Shiri} proved that the algorithm cannot be implemented in some cases, they were able to modify it without changing its general strategy. Hence, we focus only on the original algorithm and ignoring the unnecessary for our purpose details.

%The algorithm \textsc{RandBacktrack} can be seen as a randomized version of \textsc{Backtrack} and works as follows:

%Consider the $k+1$ shortest $s-t$ paths in the graph. The algorithm defines an appropriate probability distribution and then chooses a path according to this distribution that the agent tries to traverse. If this path is feasible, the algorithm terminates. If it is blocked, the agent returns to $s$ and repeats the procedure for a smaller set of paths. 

\subsection{Notations}
For every edge $e \in E$, we get a prediction on whether $e$ is blocked or not. We define the error of the predictions as the total number of false predictions we have compared to the real instance. Formally every edge has prediction error $\textsc{er}(e)\in \{0,1\}$. 

The total prediction error is given by:
\begin{equation*}
    error = \sum_{e \in E}{\textsc{er}(e)}
\end{equation*}

%\subsection{Notation}
For proving lower bounds, we will refer several times  to the graph $\mathcal{G}^*$, which means the graph in Figure~\ref{fig:graph}.   When we say that $P_i$ is blocked, we mean that the edge with cost $0$ of $P_i$ is blocked. We omit these details in the proofs for ease of explanation. Moreover, when we refer to a path we always mean an $s-t$ path in the rest of the paper. 
Note that the lower bounds we get occur with strictly positive values on the edges, simply by replacing $0$ with $\epsilon>0$ in $\mathcal{G}^*$.

We denote by $OPT$ the optimal offline cost of the $k$-CTP instance, $ALG$ the value of an algorithm under study, and by $r$ the competitive ratio of an online algorithm to avoid any confusion with the cost $c$ of an edge on the graph. More specifically, $r=\frac{ALG}{OPT}$ in the deterministic case and $r=\frac{\mathbb{E}[ALG]}{OPT}$ in the randomized one.

\section{Tradeoffs between consistency and robustness}\label{sec:tradeoffs}

In this section, we study the tradeoffs between consistency and robustness. As explained before, we express tradeoffs by answering the following question: if we want our algorithm to be $(1+\epsilon)$-competitive if the prediction is correct (consistency $1+\epsilon$), what is the best competitive ratio we can get when the prediction is not correct (robustness)?

We deal with the deterministic case in Section~\ref{subsec:tradeoffdet} and the randomized one in Section~\ref{subsec:tradeoffrand}.

\subsection{Tradeoffs for deterministic  algorithms}\label{subsec:tradeoffdet}

\begin{theorem}
\label{theor_det_e}
Any deterministic %prediction-augmented online 
algorithm that achieves competitive ratio smaller than $1+\epsilon$ when the prediction is correct, achieves a ratio of at least $2k-1+\frac{4k}{\epsilon}$ when the prediction is not correct, even when the error is at most $2$ and the graph is path-disjoint.
%The parameter $k$, here, bounds the real number of blocked edges from above.
\end{theorem}

\begin{proof}
Consider a graph $\mathcal{G}^*$ with $k+1$ \iffalse source-sink \fi paths $P_1$, $P_2$, ..., $P_k$, $P_{k+1}$, which are node-disjoint. The paths $P_1$, $P_2$,..., $P_{k-1}$, $P_{k}$ have costs equal to $1$ ($c_1 = c_2 = ... = c_k =1$) and path $P_{k+1}$ has cost $c_{k+1} = \frac{2k}{\epsilon}$. $P_1$, $P_2$,...,$P_k$ are predicted to be blocked ($k$ predicted blocked edges) and $P_{k+1}$ is feasible.

If there is no error, the predicted instance is also the real one, $P_{k+1}$ is optimal. To get a competitive ratio smaller than $1+\epsilon$ (consistency bound), a deterministic online algorithm cannot follow all paths $P_1$, $P_2$, \dots, $P_k$ before exploring $P_{k+1}$ as the ratio would be $r=\frac{2c_1}{c_{k+1}}+...+\frac{2c_k}{c_{k+1}}+1=1+\epsilon$. Therefore, $P_{k+1}$ is visited before at least one path $P_1$, $P_2$, \dots, $P_k$.

When an adversary blocks $P_{k+1}$ and all the other paths visited by the algorithm except for the last one ($k$ blocks in total), it creates a new instance with $error=2$. The optimal cost is $1$ and the algorithm now has competitive ratio:
\begin{equation*}
    r = \frac{2(k-1)+2c_{k+1}+1}{1}=2k-1+\frac{4k}{\epsilon}
\end{equation*}

Consequently, we have a Pareto lower bound  $(1+\epsilon, 2k-1+\frac{4k}{\epsilon}]$.
\end{proof}

%\be{to be put before} We say that an algorithm is $(a,b]$-competitive, when it achieves a ratio of less than $a$ when the prediction is correct and no more than $b$ otherwise. 

We now give an algorithm that matches the previous lower bound.

\textsc{E-Backtrack} is formally described in Algorithm~\ref{algo_disjdecomp}. It basically executes \textsc{Backtrack}, but interrupts at some point its execution in order to explore the shortest unblocked-predicted path. The interruption point (determined by Equation~(\ref{eq:tradeoff1}) in the description of the algorithm)  is chosen sufficiently early to ensure good consistency and not too early to ensure good robustness.  

\IncMargin{1em}
\begin{algorithm}
\caption{\textsc{E-Backtrack}}\label{algo_disjdecomp}
\SetKwData{Left}{left}\SetKwData{This}{this}\SetKwData{Up}{up}
\SetKwFunction{Union}{Union}\SetKwFunction{FindCompress}{FindCompress}
\SetKwInOut{Input}{Input}\SetKwInOut{Output}{Output}
\Input{An instance of CTP with prediction with parameter $k$, $\epsilon>0$}
\Output{An $s-t$ path}
\BlankLine
$P_{pred},c_{pred} \leftarrow$ shortest path and its cost after removing all predicted blocked edges
%\emph{special treatment of the first line}\;

Execute \textsc{Backtrack} and explore paths $P_1,\dots,P_j$, of cost $c_1,\dots,c_j$ until one of the following cases occurs: 

\hspace{0.4cm} (a) $t$ is reached

\hspace{0.4cm} (b) the next path $P_{j+1}$ to explore is such that:
    \begin{equation}\label{eq:tradeoff1}
            2c_1+2c_2+...+2c_j +2c_{j+1} \geq \epsilon \cdot c_{pred}
        \end{equation}

\lIf{(a) occurs}{
    Return the found path}  
\Else{%{\it ((b) occurs)}    

Explore $P_{pred}$ (if not yet known to be blocked)
    
    \lIf{$P_{pred}$ is not blocked}{output it}
    \lElse{Resume the execution of \textsc{Backtrack}}
}
\end{algorithm}\DecMargin{1em}

\begin{theorem}
\label{alg_det_e}
For $0<\epsilon\leq  2k$, \textsc{E-Backtrack} is a deterministic $(1+\epsilon, 2k-1+\frac{4k}{\epsilon}]$-competitive algorithm.
\end{theorem}
\begin{proof}
We denote by $ALG$ the cost of algorithm \textsc{E-Backtrack}.

Suppose first that case $(a)$ occurs. Then $j\leq k+1$ as there are at most $k$ blocked edges, and $2c_1+\cdots+2c_j< \epsilon\cdot c_{pred}$ (otherwise case $(b)$ would have occurred earlier). In the case where the prediction is correct, $OPT=c_{pred}$ and, using 
$\epsilon\leq 2k$:
$$ALG\leq 2c_1+\dots+2c_{j-1}+c_j< \epsilon\cdot c_{pred}<(1+\epsilon)c_{pred}$$ 
If the prediction is not correct, then $OPT=c_j$ and:  $$ALG\leq 2c_1+\dots+2c_{j-1}+c_j\leq (2k+1)OPT \leq (2k-1+4k/\epsilon) OPT$$

Suppose now that case $(b)$ occurs. As explained earlier $2c_1+\cdots+2c_j< \epsilon\cdot c_{pred}$. In the case where the prediction is correct, $OPT=c_{pred}$ and: 
$$ALG\leq 2c_1+\dots+2c_{j-1}+2c_j+c_{pred} <(1+\epsilon)c_{pred}$$

In the case where the prediction is not correct, if $P_{pred}$ were already known to be blocked, then we directly get a ratio $2k+1\leq 2k-1+4k/\epsilon$. Otherwise, let $P_1,\dots,P_j,P_{pred},P_{j+1},\dots,P_t$ be the paths explored by \textsc{E-Backtrack}. As there are at most $k$ blocked edges, $t\leq k$ (note that the exploration of $P_{pred}$ does give a previously unknown blocked edge). Moreover, $c_1\leq c_2\leq \dots \leq c_t=OPT$. We get: \begin{equation}\label{eq:tradeoff2}
ALG= 2\sum_{i=1}^{t-1}c_i+2c_{pred}+c_t\leq  (2k-1)OPT+2c_{pred}\end{equation}
Using~(\ref{eq:tradeoff1}) we know that $c_{pred}\leq 2\sum_{i=1}^{j+1}c_i/\epsilon\leq 2\sum_{i=1}^{t}c_i/\epsilon\leq 2k c_{t}/\epsilon $.
Then Equation~(\ref{eq:tradeoff2}) gives:
\begin{equation*}
ALG\leq   \left(2k-1+\frac{4k}{\epsilon}\right)OPT\end{equation*}
\end{proof}

\subsection{Randomized bounds and algorithms}\label{subsec:tradeoffrand}

We now consider the randomized case. As explained in the introduction, we restrict ourselves to the path-disjoint graphs for randomized algorithms. %The proof of this first result can be found in Appendix~\ref{app:theor_rand_e}.

\begin{theorem}
\label{theor_rand_e}
Any randomized %prediction-augmented online 
algorithm that achieves competitive ratio at most $1+\epsilon$ when the prediction is correct, achieves a ratio of at least $k+\frac{k}{\epsilon}$ when the prediction is not correct, even when the error is at most $2$ and the graph is path-disjoint.
\end{theorem}

\begin{proof}
Consider a graph $\mathcal{G}^*$ with $k+1$ \iffalse source-sink \fi paths $P_1$, $P_2$, ..., $P_k$, $P_{k+1}$, which are node-disjoint. The paths $P_1$, $P_2$,..., $P_{k-1}$, $P_{k}$ have costs equal to $1$ ($c_1 = c_2 = ... = c_k =1$) and path $P_{k+1}$ has cost $c_{k+1} = \frac{k}{\epsilon}$. $P_1$, $P_2$,...,$P_k$ are predicted to be blocked ($k$ predicted blocked edges) and $P_{k+1}$ is feasible.

If there is no error, the predicted instance is also the real one, $P_{k+1}$ is optimal ($OPT=c_{k+1}=\frac{k}{\epsilon}$) and the competitive ratio must be at most $1+\epsilon$. In the above instance, any deterministic algorithm can achieve one of the following competitive ratios:

\begin{itemize}
    \item $r = 1$, when choosing only $P_{k+1}$.
    
    \item $r = \frac{2 \epsilon}{k}+1$, when choosing paths $P_i$, $P_{k+1}$ with $i\ne k+1$.
    
    \item $r = \frac{4 \epsilon}{k}+1$, when choosing paths $P_i$, $P_j$, $P_{k+1}$ with $i, j\ne k+1$ and $i\ne j$.
    
    \item  \quad \enspace .\enspace.\enspace.
   
    \item $r = \frac{2k\cdot \epsilon}{k}+1$, when choosing all $k$ paths $P_1$,..., $P_k$ and then $P_{k+1}$.   
\end{itemize} 

A randomized algorithm can be viewed as a probability distribution over all deterministic algorithms. We assume that an arbitrary randomized algorithm chooses with cumulative probability $p_1$ the deterministic algorithms that achieve a ratio of $1$ (here there is only one algorithm), with cumulative probability $p_2$ the deterministic algorithms that achieve a ratio of $\frac{2 \epsilon}{k}+1$, and so on. We also have that:
\begin{equation}
\label{eq5}
    \sum_{i=1}^{k+1}p_i=1
\end{equation}

Hence, the competitive ratio of an arbitrary randomized algorithm is:
\begin{equation*}
    r = p_1\cdot 1 + p_2  \bigg( \frac{2 \epsilon}{k}+1 \bigg) + p_3 \bigg( \frac{4 \epsilon}{k}+1 \bigg) + ... + 
    p_{k+1} \bigg( \frac{2k\cdot \epsilon}{k}+1 \bigg)
\end{equation*}

Since $r \leq 1+\epsilon$, we have that:
\begin{equation*}
    p_1 + p_2 \bigg( \frac{2 \epsilon}{k}+1 \bigg) + p_3 \bigg( \frac{4 \epsilon}{k}+1 \bigg) + ... + 
    p_{k+1} \bigg( \frac{2k\cdot \epsilon}{k}+1 \bigg) \leq 1 + \epsilon
\end{equation*}
\begin{equation*}
    \Rightarrow \sum_{i=1}^{k+1}p_i + p_2 \cdot \frac{2 \epsilon}{k}  + p_3 \cdot  \frac{4 \epsilon}{k} + ... + 
    p_{k+1} \cdot \frac{2k\cdot \epsilon}{k} \leq 1 + \epsilon
\end{equation*}
From (\ref{eq5}) it follows that:
\begin{equation*}
   p_2 \cdot  \frac{2 \epsilon}{k}  + p_3 \cdot  \frac{4 \epsilon}{k} + ... + 
    p_{k+1} \cdot \frac{2k\cdot \epsilon}{k} \leq  \epsilon
\end{equation*}
\begin{equation}
   \Rightarrow p_2  + 2 p_3 + ... + 
    k\cdot p_{k+1}  \leq  \frac{k}{2} \label{eq:brpotr}
\end{equation}

\begin{comment}
Using (\ref{eq5}) $k$ times we get:
\begin{equation*}
    (1-p_1)+(1-p_1-p_2)+...+(1-p_1-p_2-...-p_k) \leq \frac{k}{2}
\end{equation*}
\begin{equation*}
   \Rightarrow k-k\cdot p_1-(k-1) p_2-...-p_k \leq \frac{k}{2}
\end{equation*}
\begin{equation}
\label{eq6}
   \Rightarrow k\cdot p_1+(k-1) p_2+...+p_k \geq \frac{k}{2}
\end{equation}
\end{comment}

%When 
We now look at the case that the prediction is wrong. Consider the (randomized) set of instances where %an oblivious adversary chooses
 the path $P_i$ is unblocked, where $i$ is chosen uniformly at random in $\{1,...,k \}$, and path $P_{k+1}$ is blocked. Note that these instances have $error=2$.
So, only path $P_i$ is feasible with $cost=1$ and $OPT=1$. Consider a deterministic algorithm which explores (until it finds an unblocked path) $\ell\geq 0$ paths among $P_1,\dots,P_k$, then $P_{k+1}$, then the remaining paths among the first $k$. On the previously given randomized set of instances, it will explore $P_{k+1}$ with probability $(1-\ell /k)$, and will find the unblocked path after exactly $t$ explorations with probability $1/k$ (for any $t$). Thus, the expected cost of such an algorithm on the considered randomized set of instances is: $$\mathbb{E}_\ell=\left(1-\frac{\ell}{k}\right)2c_{k+1}+ \frac{(1+3+\dots+(2k-1))}{k}=\left(1-\frac{\ell}{k}\right)2c_{k+1}+k$$
Then, the expected cost of the randomized algorithm (which chooses such an algorithm with probability $p_{\ell+1}$) on the given randomized set of instances verifies:
\begin{equation*}
\mathbb{E}[ALG]\geq \sum_{\ell=0}^k p_{\ell+1}\mathbb{E}_\ell=\sum_{\ell=0}^k p_{\ell+1} \left(\left(1-\frac{\ell}{k}\right)2c_{k+1}+k\right)  
\end{equation*}
\begin{equation*}
= 2c_{k+1}\left(1-\frac{\sum_{\ell=0}^kp_{\ell+1}\ell}{k}\right)+k 
\end{equation*}

Equation~(\ref{eq:brpotr}) gives $\sum_{\ell=0}^kp_{\ell+1} \ell \leq k/2$, so we have $\mathbb{E}[ALG]\geq  c_{k+1}+k=k+\frac{k}{\epsilon}$.
\end{proof}

%We remark that the above bounds are valid, even if the graph is path-disjoint.

We now give a randomized algorithm that is $[1+\epsilon, k+\frac{4k}{\epsilon}]$-competitive. Similarly as \textsc{E-Backtrack}, it executes \textsc{RandBacktrack} but interrupts at some point (determined by Equation (\ref{eq:tradeoff3})) its execution in order to explore the shortest unblocked-predicted path.

\begin{theorem}
\label{alg_rand_e}
For $0<\epsilon\leq  k$, \textsc{E-RandBacktrack} is a randomized $[1+\epsilon, k+\frac{4k}{\epsilon}]$-competitive algorithm.
\end{theorem}

\begin{proof}

We denote by $ALG$ the cost of algorithm 
\textsc{E-RandBacktrack} and by $A_{k-1}$ the cost of \textsc{RandBacktrack} (both $ALG$ and $A_{k-1}$ are random variables).  

The proof is based on the following observations.

\vspace{0.2cm}

\noindent {\it Observation 1. At the time of the algorithm when $(a)$, $(b)$ or $(c)$ occurs, $TVL \leq \epsilon \cdot c_{pred}$. In particular, if $(a)$ or $(b)$ occurs, $A_{k-1}\leq \epsilon \cdot c_{pred}$.}

\vspace{0.2cm}

\noindent Indeed, otherwise case $(c)$ would have occurred earlier. 

\vspace{0.2cm}

\noindent{\it Observation 2. If $P_1,\dots,P_k$ are blocked, then $P_{pred}$ is unblocked and optimal. In particular, if $(b)$ occurs then $P_{pred}$ is unblocked and optimal.}

\vspace{0.2cm}

\noindent Indeed, if  $P_1,\dots,P_k$ are blocked, there is no other blocked path (as $k$ upper bounds the number of blocked paths) so $P_{pred}$ is unblocked. If the set of $k$ blocked edges are exactly the predicted ones, then $P_{pred}$ is by definition optimal. Otherwise, one path $P_i$ is not predicted to be blocked (as $k$ upper bounds the number of predicted blocked paths), hence $c_{pred}\leq c_i\leq c_k$. But by definition of paths $P_1,\dots,P_k$, if they are all blocked then $OPT\geq c_k$, so again $OPT = c_{pred}$.

\IncMargin{1em}
\begin{algorithm}[H]
\SetKwData{Left}{left}\SetKwData{This}{this}\SetKwData{Up}{up}
\SetKwFunction{Union}{Union}\SetKwFunction{FindCompress}{FindCompress}
\SetKwInOut{Input}{Input}\SetKwInOut{Output}{Output}
\Input{An instance of CTP with prediction with parameter $k$, $\epsilon>0$}
\Output{An $s-t$ path}
\BlankLine
$P_{pred},c_{pred} \leftarrow$ shortest path and its cost after removing all predicted blocked edges
%\emph{special treatment of the first line}\;

$P_1,\dots,P_k$ of cost $c_1,\dots,c_k \leftarrow$ $k$ shortest paths except for $P_{pred}$ \footnotemark[2]

$\textsc{TVL} \leftarrow$ total visited length of \textsc{RandBacktrack} before exploring the next path

Execute \textsc{RandBacktrack} on paths $P_1,\dots,P_k$ with parameter $k-1$ until one of the following cases occurs: 

\hspace{0.4cm} (a) $t$ is reached

\hspace{0.4cm} (b) $t$ is not reached and \textsc{RandBacktrack} terminates

\hspace{0.4cm} (c) the next path $P_{new}$ of cost $c_{new}$ to explore is such that:
    \begin{equation}\label{eq:tradeoff3}
            \textsc{TVL} +2c_{new} > \epsilon \cdot c_{pred}
        \end{equation}

\lIf{(a) occurs}{
    Return the found path} 
\ElseIf{(b) occurs}{
    Explore $P_{pred}$ and output it
}
\Else{
    Explore $P_{pred}$ 
    
    \lIf{$P_{pred}$ is not blocked}{output it}
    \lElse{Resume the execution of \textsc{RandBacktrack}}
}

\caption{\textsc{E-RandBacktrack}}\label{rand_algo_disjdecomp}
\end{algorithm}\DecMargin{1em}

\footnotetext[2]{If the graph contains less than $k$ disjoint paths, then choose the maximum number of paths $l < k$ and run \textsc{RandBacktrack} with parameter $l-1$. The analysis remains the same.}

\vspace{0.2cm}

\noindent{\it Observation 3. If $(c)$ occurs, then $A_{k-1}>\frac{\epsilon \cdot c_{pred}}{2}$.}

\vspace{0.2cm}

\noindent Indeed, when $(c)$ occurs \textsc{RandBacktrack} has cost at least $TVL +c_{new}>\epsilon \cdot c_{pred}/2$.\\

Then, suppose first that $P_{pred}$ is unblocked and optimal. Following observation 1, if $(a)$ or $(b)$ occurs we have $ALG\leq A_{k-1}+c_{pred}\leq (1+\epsilon)c_{pred}$, and in case $(c)$ also $ALG\leq(1+\epsilon)c_{pred}$. So anyway $ALG\leq (1+\epsilon)OPT$, and in particular $\mathbb{E}[ALG]\leq (1+\epsilon)OPT$. So \textsc{E-RandBacktrack} is $(1+\epsilon)$-competitive when there is no error. As $(1+\epsilon)\leq k+4k/\epsilon$ ($\epsilon\leq k$), the robustness bound is also verified in this case.

Now, suppose that we are in the other case, i.e.,  $P_{pred}$ is either blocked, or unblocked but not optimal.  Note that the prediction is not correct here. Following Observation 2, $(b)$ cannot occur, and one path in $P_1\dots,P_k$ is unblocked (so \textsc{RandBacktrack} does find a path before terminating). Then in case $(a)$ $ALG\leq A_{k-1}$, and in case $(c)$, anyway, $ALG\leq A_{k-1}+2c_{pred}$. So we get:
\begin{equation}
    \mathbb{E}[ALG]\leq \mathbb{E}[A_{k-1}]+2c_{pred}\cdot Pr(c)\label{equtoto}
\end{equation}    
    \noindent where $Pr(c)$ denotes the probability that case $c$ occurs. Following Observation 3, if $(c)$ occurs $A_{k-1}>\epsilon \cdot c_{pred}/2$. Using Markov Inequality, we have:
\begin{equation}
Pr(c)\leq Pr\left(A_{k-1}>\frac{\epsilon \cdot c_{pred}}{2}\right)\leq \frac{2\mathbb{E}[A_{k-1}]}{\epsilon \cdot c_{pred}} \label{equtoto2}
\end{equation}
Using Equations~(\ref{equtoto}) and (\ref{equtoto2}) we get $\mathbb{E}[ALG]\leq \mathbb{E}[A_{k-1}]+\frac{4\mathbb{E}[A_{k-1}]}{\epsilon}$. As (at least) one path is unblocked in $P_1,\dots,P_k$, $\mathbb{E}[A_{k-1}]\leq k \cdot OPT$, and the result follows. 

\end{proof}

The guarantee provided by \textsc{E-RandBacktrack} (Theorem \ref{alg_rand_e}) does not exactly match the lower bound of Theorem~\ref{theor_rand_e}. While closing this gap is left as an open question, we conjecture that the exact (optimal) tradeoff corresponds to the lower bound $[1+\epsilon,k+\frac{k}{\epsilon}]$ (for path-disjoint graphs). Towards this conjecture, we present two cases where we have an upper bound that matches  the $[1+\epsilon,k+\frac{k}{\epsilon}]$-lower bound. The two special cases (proved in Appendices~\ref{app:ubk=1} and~\ref{app:rbu}) are $k=1$ (Theorem~\ref{th:ubk=1}) and the case of uniform costs (Theorem~\ref{th:RBU}).  

For the case $k=1$, the algorithm (fully described in Appendix~\ref{app:ubk=1}) is a modified version of \textsc{RandBacktrack} with different probabilities, specifically settled exploiting the fact that there is (at most) one blocked edge and one predicted blocked edge. 

For the case of uniform cost, the algorithm (fully described in  Appendix~\ref{app:rbu}) explores at first a path with an appropriate probability. If the path is blocked, it then executes \textsc{RandBacktrack} on the remaining graph. 

\begin{theorem}\label{th:ubk=1}
For $0<\epsilon\leq  1$, there exists a randomized $[1+\epsilon, 1+\frac{1}{\epsilon}]$-competitive algorithm when the graph is path-disjoint and $k=1$. 
\end{theorem}

%We now give a randomized algorithm \textsc{RandBacktrackU} (see~Algorithm~\ref{algo_uniform}) that matches the $[1+\epsilon, 1+\frac{k}{\epsilon}]$-lower bound when $k$ is considered to be known, and all paths are of uniform cost and node-disjoint. It basically explores at first a path with an appropriate probability. If the path is blocked, it then simply executes \textsc{RandBacktrack} on the remaining graph.

%We now deal with the case of uniform costs. The algorithm explores at first a path with an appropriate probability. If the path is blocked, it then executes \textsc{RandBacktrack} on the remaining graph. The algorithm and the proof of the Theorem can be found in Appendix~\ref{app:rbu}.

\begin{theorem}\label{th:RBU}
For $0<\epsilon\leq  k$, there exists a randomized $[1+\epsilon, k+\frac{k}{\epsilon}]$-competitive algorithm when the graph is path-disjoint and the costs are uniform.
\end{theorem}

\section{Robustness analysis}\label{sec:robustness}

This section is devoted to the analysis of competitive ratios that can be achieved depending on the error made in the prediction. Section~\ref{subsec:detb} deals with deterministic bounds, and Section~\ref{subsec:randb} with randomized ones.

\subsection{Deterministic bounds}\label{subsec:detb}

As a first result, we show that the lower bound of $2k+1$ on (deterministic) competitive ratios still holds in our model with prediction even if the prediction error is (at most) 2 (see Appendix~\ref{app:2k+1} for the proof).

\begin{theorem}
\label{theorem1}
There is no deterministic  algorithm  that
achieves competitive ratio smaller than $2k+1$, even when the prediction has error at most $2$ and the graph is path-disjoint.
\end{theorem}

We can easily achieve a matching $2k+1$ upper bound using the optimal deterministic algorithm \textsc{Backtrack} ignoring the predictions completely. Hence, the lower bound is tight.

\begin{comment}
\be{I am not sure we shall include the next result, at least in a conference version.}
\begin{theorem}
\label{theorem_three}
If the number of predicted blocked edges $k_p$ is greater than the number $\kappa$ of the real ones, then there is no deterministic online algorithm with predictions that achieves competitive ratio smaller than $2\kappa+1$, even when the prediction has error at most $1$. 
%The parameter $\kappa$ bounds the number of real blocked edges from above.
\end{theorem}

\begin{proof}
Consider a graph $\mathcal{G}^*$ with $\kappa+1$ \iffalse source-sink \fi paths $P_1$, $P_2$,..., $P_{\kappa}$, $P_{\kappa+1}$, which are node-disjoint. All the paths have equal costs, meaning that $c_1=c_2=...=c_{\kappa}=c_{\kappa+1}$.  All paths $P_1$, $P_2$,..., $P_{\kappa}$, $P_{\kappa+1}$ are predicted to be blocked ($\kappa+1$ predicted blocks).

Every deterministic algorithm corresponds to a permutation,
which describes in which order the paths are being explored. The adversary can only create an instance with $error$ no more than $1$.

The adversary blocks every path tried by a deterministic algorithm except for the last one ($\kappa$ blocks). Then the adversary creates an instance with a prediction error of at most $1$ and the algorithm has competitive ratio at least $2\kappa+1$.
\end{proof}
\end{comment}

We now consider the remaining case, when the error is (at most) 1. In this case we show that an improvement can be achieved with respect to the $2k+1$ bound. More precisely, we first show in Theorem~\ref{theorem2k-1} that for any $k\geq 3$ a ratio $2k-1$ can be achieved, and in Theorem~\ref{theoremUBk=2} that a ratio $\frac{3+\sqrt{17}}{2}\simeq 3.56$ can be achieved for $k=2$. We show in Theorems~\ref{theorem2} and~\ref{theoremLBk=2} respectively that these bounds are tight. Theorem~\ref{theor_k_1} settles the case $k=1$.

Theorem~\ref{theorem2k-1} is proven in Appendix~\ref{app:blabla}. The two main ingredients of the claimed algorithm are (1) a careful comparison of the lengths of the shortest path and the shortest path without predicted blocked edges, to decide which one to explore, and (2) the fact that when the error is at most 1, if a blocked edge is discovered and was not predicted to be so, then we know exactly the set of blocked edges (i.e., the predicted ones and the new one) and thus we can determine directly the optimal solution without further testing.

\begin{theorem}\label{theorem2k-1}
There is a $(2k-1)$-competitive algorithm when the prediction error is at most $1$ and $k\geq 3$ is known. 
%The parameter $k$ bounds the real and the predicted number of blocked edges from above.
\end{theorem}

We now show that the upper bound of $2k-1$ (for $k\geq 3$) is tight (see Appendix~\ref{app:theorem2} for the proof).

\begin{theorem}
\label{theorem2}
There is no deterministic algorithm  that 
achieves competitive ratio smaller than $2k-1$, even when the prediction has $error$ at most $1$ and the graph is path-disjoint. 
\end{theorem}

We next examine separately the cases for $k=1$ (Theorem~\ref{theor_k_1}, proved in Appendix~\ref{app:theor_k_1}) and $k=2$.

\begin{theorem}
\label{theor_k_1}
When $k=1$, there is no deterministic  algorithm   that achieves competitive ratio smaller than $3$, even when the prediction has error at most $1$ and the graph is path-disjoint.
\end{theorem}

The previous $3$-lower bound is clearly tight using the optimal deterministic algorithm \textsc{Backtrack} ($2\cdot1+1=3$). For $k=2$, we have matching lower and upper bounds of $\frac{3+\sqrt{17}}{2}\simeq 3.56$. The proof of the lower bound (Theorem~\ref{theoremLBk=2}) is in Appendix~\ref{app:LBk=2}. The upper bound (Theorem~\ref{theoremUBk=2}) follows from an adaptation of \textsc{Err1-Backtrack} for the case $k=2$ 
(the description of Algorithm \textsc{Err1-Backtrack2} and the proof of the theorem are in Appendix~\ref{app:theoremUBk=2}).

\begin{theorem}\label{theoremLBk=2}
When $k=2$, there is no deterministic algorithm that achieves competitive ratio smaller than $\frac{3+\sqrt{17}}{2}$, even when the prediction has error at most $1$ and the graph is path-disjoint.
\end{theorem}

%Algorithm \textsc{Err1-Backtrack2} which refers to the case of error at most 1 with $k=2$, is described in Algorithm~\ref{algo_err1BT2} in a recursive form. The algorithm is a just simple modification of \textsc{Err1-Backtrack} when $k \leq 3$ and the analysis is similar.

\begin{theorem}\label{theoremUBk=2}
There exists a $\frac{3+\sqrt{17}}{2}$-competitive algorithm when the prediction error is at most $1$ and $k=2$.
\end{theorem}

\subsection{Randomized bounds}\label{subsec:randb}

We now consider randomized algorithms, and as explained in introduction we focus on graphs with node disjoint paths. Similarly as for the deterministic case, we first show that the lower bound of $k+1$ on (randomized) competitive ratios still holds in our model with prediction even if the prediction error is (at most) 2. %(see Appendix~\ref{app:nokeeps} for the proof).

\begin{theorem}
\label{theor_rand_k1}
There is no randomized algorithm  that 
achieves competitive ratio smaller than $k+1$ against an oblivious adversary, even when the prediction has error at most $2$ and the graph is path-disjoint. 
%The parameter $k$, here, bounds the real number of blocked edges from above.
\end{theorem}

\begin{proof}
In what follows we provide a randomized set of instances on which the expected cost of any deterministic algorithm is at least $k+1$ times the optimal cost. It follows from Yao's Principle~\cite{Yao} that the competitive ratio of any randomized algorithm is at least $k+1$.

Consider a graph $\mathcal{G}^*$ with $k+1$ \iffalse source-sink \fi paths $P_1$, $P_2$, ..., $P_k$, $P_{k+1}$, which are node-disjoint. All the paths have costs equal to $1$, meaning that $c_1=c_2=...=c_k=c_{k+1}=1$.

Paths $P_1$, $P_2$, ... , $P_k$ are predicted to be blocked. We choose $i \in \{1,...,k+1 \}$ uniformly at random and block all paths $P_j$ with $j \ne i$ ($k$ blocked edges as all paths are node-disjoint). The prediction has an error of at most $2$.

So, only the path $P_i$ is feasible at cost $1$ and the optimal offline cost is $1$. Furthermore, an arbitrary deterministic online algorithm finds path $P_i$ on the $l$th trial for $l=1,...,k+1$ with probability $\frac{1}{k+1}$.

If the algorithm is successful on its $l$th try, it incurs a cost of $2l-1$, and thus it has an expected cost of at least
    $ \frac{1}{k+1} \sum_{l=1}^{k+1}(2l-1) = \frac{1}{k+1} \cdot (k+1)^{2} = k+1$.
%\end{equation*}
\end{proof}

If $k$ is known and the graph is path-disjoint, then the above $k+1$ lower bound is tight using the optimal randomized online algorithm \textsc{RandBacktrack} without predictions.

We finally consider the case of error (at most) $1$. We show in Theorem~\ref{th:randk} (see Appendix~\ref{app:noke1} for the proof) a lower bound of $k$. We leave as an open question closing this gap between $k$ and $k+1$ when the error is at most $1$ and $k \geq 2$. For the special case of $k=1$, it is easy to show a matching lower bound of $2$.

\begin{theorem}\label{th:randk}
There is no randomized algorithm  that 
achieves competitive ratio smaller than $k$ against an oblivious adversary, even when the prediction has error at most $1$ and the graph is path-disjoint. 
\end{theorem}

\bibliographystyle{splncs04}
\bibliography{mybibliography.bib}

\newpage

\appendix

\section{Missing proofs from Section~\ref{sec:tradeoffs}}\label{app:a}

%\subsection{Proof of Theorem~\ref{theor_rand_e}}\label{app:theor_rand_e}

%{\bf Theorem~\ref{theor_rand_e}.}
%{\it 
%Any randomized %prediction-augmented online 
%algorithm that achieves competitive ratio at most $1+\epsilon$ when the prediction is correct, achieves a ratio of at least $k+\frac{k}{\epsilon}$ when the prediction is not correct, even when the error is at most $2$ and the graph is path-disjoint.
%} 

\subsection{Proof of Theorem~\ref{th:ubk=1}}\label{app:ubk=1}

{\bf Theorem~\ref{th:ubk=1}.}{\it
For $0<\epsilon\leq  1$, there exists a randomized $[1+\epsilon, 1+\frac{1}{\epsilon}]$-competitive algorithm when the graph is path-disjoint and $k=1$. 
} 

\begin{proof}
Let $P_{pred}$ be the shortest path on the graph that remains after removing the only predicted block (if any) and $P_1$ be the shortest path with $P_1 \ne P_{pred}$. There is at least one feasible path among them because $1$ is an upper bound on the real number of blocked edges.

Let us consider the following algorithm \textsc{RandBacktrackOne} which is a modified version of \textsc{RandBacktrack} with different probabilities. 

\begin{itemize}
    \item The algorithm first chooses a path and tries to traverse it. If this path is feasible, the algorithm terminates. If it is blocked, we return to $s$ and traverse the other path.
    
    \item If $c_{pred} > \frac{2}{\epsilon} \cdot c_1$: it chooses at first path $P_1$.
    
    \item Otherwise: it chooses at first path $P_{pred}$ with probability 
    $p_{pred}=\frac{2c_1-\epsilon \cdot c_{pred}}{2c_1}$ and path $P_1$ with probability $p_1 = 1-p_{pred} = \frac{\epsilon \cdot c_{pred}}{2c_1}$.
\end{itemize}

When $error=0$, then $P_{pred}$ is the optimal path ($OPT = c_{pred}$).
\begin{itemize}
    \item If $c_{pred} > \frac{2}{\epsilon} \cdot c_1$, then the algorithm chooses at first $P_1$ and then $P_{pred}$.
    
    \textsc{RandBacktrackOne} has competitive ratio:
\begin{equation*}
    r \leq \frac{2c_1+c_{pred}}{c_{pred}} <1+\epsilon
\end{equation*}
    \item Otherwise:  $c_{pred} \leq \frac{2}{\epsilon} \cdot c_1$ and the ratio is:
    \begin{equation*}
         r \leq p_1  \bigg( \frac{2c_1}{c_{pred}}+1\bigg) +
         p_{pred}\cdot 1
    \end{equation*}
    \begin{equation*}
         = p_1 \cdot \frac{2c_1}{c_{pred}}+1=
         \frac{\epsilon \cdot c_{pred}}{2c_1} \cdot \frac{2c_1}{c_{pred}} +1 
    \end{equation*}
    \begin{equation*}
         = 1 + \epsilon 
    \end{equation*}
\end{itemize}

When $error \ne 0$, then:

\begin{itemize}
    \item If $P_{pred}$ is optimal, then the analysis is the same as when $error=0$ and since $1+\epsilon \leq 1+ \frac{1}{\epsilon}$ (as $\epsilon \leq 1$) we get that: 
    \begin{equation*}
        r \leq 1+ \frac{1}{\epsilon}
    \end{equation*}
    
    \item If $P_1$ is optimal ($OPT = c_1$):
    \begin{itemize}
        \item If $c_{pred} > \frac{2}{\epsilon} \cdot c_1$, then the algorithm chooses at first $P_1$ and its competitive ratio is $1$.
        
        \item Otherwise:  $c_{pred} \leq \frac{2}{\epsilon} \cdot c_1$ and \textsc{RandBacktrackOne} has ratio:
        \begin{equation*}
             r \leq p_1 \cdot 1 + p_{pred}  \bigg( \frac{2c_{pred}}{c_1}+1\bigg)
        \end{equation*}
        \begin{equation*}
             = 1 + p_{pred} \cdot \frac{2c_{pred}}{c_1}
        \end{equation*}
        \begin{equation*}
             = 1 + \frac{2c_1-\epsilon \cdot c_{pred}}{2c_1} \cdot \frac{2c_{pred}}{c_1}
        \end{equation*}
        \begin{equation*}
             = 1 + \frac{2c_1-\epsilon \cdot c_{pred}}{c_1} \cdot \frac{c_{pred}}{c_1}
        \end{equation*}
        We have that:
        \begin{equation*}
             \frac{2c_1 \cdot c_{pred}-\epsilon \cdot c_{pred}^2}{c_1^2} \leq \frac{1}{\epsilon}
        \end{equation*}  
        \begin{equation*}
            \Leftrightarrow 2 \epsilon \cdot c_1 \cdot c_{pred}-\epsilon^2 \cdot c_{pred}^2 \leq c_1^2
        \end{equation*}
        \begin{equation*}
            \Leftrightarrow (c_1- \epsilon \cdot c_{pred})^2 \geq 0 \text{, which is always true.}
        \end{equation*}
        Therefore, we get:
        \begin{equation*}
             r \leq 1 + \frac{1}{\epsilon}
        \end{equation*}
    \end{itemize}
\end{itemize}

As a result, our algorithm is $[1+\epsilon, 1+\frac{1}{\epsilon}]$-competitive.
\end{proof}

\subsection{Proof of Theorem~\ref{th:RBU}}\label{app:rbu}

{\bf Theorem~\ref{th:RBU}.} {\it 
For $0<\epsilon\leq  k$, there exists a randomized $[1+\epsilon, k+\frac{k}{\epsilon}]$-competitive algorithm when the graph is path-disjoint and the costs are uniform.}

\begin{proof}
The algorithm \textsc{RandBacktrackU} is given in Algorithm~\ref{algo_uniform}. 

\IncMargin{1em}
\begin{algorithm}
\SetKwData{Left}{left}\SetKwData{This}{this}\SetKwData{Up}{up}
\SetKwFunction{Union}{Union}\SetKwFunction{FindCompress}{FindCompress}
\SetKwInOut{Input}{Input}\SetKwInOut{Output}{Output}
\Input{An instance of CTP with prediction with parameter $k$, $\epsilon>0$}
\Output{An $s-t$ path}
\BlankLine
$P_{pred},c \leftarrow$ shortest path and its cost after removing all predicted blocked edges

$P_1, \cdots, P_k$ of equal cost $c \leftarrow$ $k$ shortest paths except for $P_{pred}$

Explore at first path $P_i$ with probability $p_i$, where:
\begin{equation}
\label{probabilities}
  p_i =
    \begin{cases}
      \frac{\epsilon}{k(k+1)}, & \text{if $i \in \{ 1, 2,..., k\}$ }\\
      \frac{k+1-\epsilon}{k+1}, & \text{if $i = pred$}\\
    \end{cases}       
\end{equation}
%\emph{special treatment of the first line}\;

\lIf{$P_i$ is not blocked}{
    output it}  
\Else{%(\textit{$P_i$ is blocked})

    Execute \textsc{RandBacktrack} with parameter $k-1$ on the graph that remains after deleting path $P_i$.
}
\caption{\textsc{RandBacktrackU}}\label{algo_uniform}
\end{algorithm}\DecMargin{1em}

We denote by $ALG$ the cost of algorithm \textsc{RandBacktrackU}, by $cost(P_i)$ the cost of \textsc{RandBacktrackU} when the algorithm explores $P_i$ at first and by $A_{k-1}$ the cost of \textsc{RandBacktrack}.

The parameter $k$ is an upper bound on the real number of blocked edges, and thus there is at least one feasible path among $P_1, \cdots, P_k, P_{pred}$. \textsc{RandBacktrack} always runs on a graph with up to $k-1$ blocked paths (as a blocked path is already discovered). Therefore, its expected cost is:
\begin{equation*}
    \mathbb{E}[A_{k-1}] \leq k\cdot c_{opt} = k \cdot c
\end{equation*}
We also note that:
\begin{equation}
\label{eq_one}
    p_1+p_2+...+p_k+p_{pred}=1
\end{equation}

In the case where the prediction is correct, $P_{pred}$ is optimal ($OPT=c$) and
\begin{equation*}
    \mathbb{E}[ALG] \leq p_1\cdot \mathbb{E}[cost(P_1)] +
    \cdots +
    p_k\cdot \mathbb{E}[cost(P_k)] + p_{pred}\cdot c
\end{equation*}
\begin{equation*}
    \leq p_1 \big( 2c + \mathbb{E}[A_{k-1}] \big) +
    ... +
    p_k \big( 2c + \mathbb{E}[A_{k-1}] \big) + p_{pred}\cdot c
\end{equation*}
\begin{equation*}
    \leq p_1 ( 2c + k\cdot c) +
    ... +
    p_k ( 2c + k\cdot c) + p_{pred}\cdot c
\end{equation*}
\begin{equation*}
    =  (k+2) (p_1+p_2+...+p_k) c + p_{pred}\cdot c
\end{equation*}
From (\ref{eq_one}) we get:
\begin{equation*}
    \mathbb{E}[ALG] \leq \big[ (k+2)(1-p_{pred}) + p_{pred} \big]  c
\end{equation*}
\begin{equation*}
    = \big[ (k+1) (1-p_{pred}) + 1 \big]  c
\end{equation*}
From (\ref{probabilities}) it follows that:
\begin{equation*}
     \mathbb{E}[ALG] \leq \bigg[ (k+1) \bigg(1- \frac{k+1-\epsilon}{k+1} \bigg)+1 \bigg]  c
\end{equation*}
\begin{equation*}
    \leq (1+\epsilon) c =  (1+\epsilon) c_{pred} 
\end{equation*}

Otherwise: we distinguish the two following cases:
\begin{itemize}
    \item If $P_{pred}$ is optimal ($OPT=c$), then the analysis is as previous and since $\epsilon \leq k$:
    \begin{equation*}
        \mathbb{E}[ALG] \leq (1+\epsilon)  c \leq (k+1) OPT 
        \leq \bigg(k+\frac{k}{\epsilon}\bigg)  OPT
    \end{equation*}
    
    \item Else: there is $i^* \in \{ 1, 2, ..., k\}$ such that $P_{i^*}$ is optimal. We also note that $OPT= cost(P_{i^*})=c_{i^*} = c$ and it follows:
    \begin{equation*}
    \mathbb{E}[ALG] \leq
    \end{equation*}
    \begin{equation*}
    \leq p_1\cdot \mathbb{E}[cost(P_1)] +
    \cdots + p_{i^*}\cdot c+\cdots+
    p_k\cdot \mathbb{E}[cost(P_k)] + p_{pred}\cdot \mathbb{E}[cost(P_{pred})]
    \end{equation*}
    \begin{equation*}
    \leq p_1 \big( 2c + \mathbb{E}[A_{k-1}] \big) +
    \cdots + p_{i^*}\cdot c + \cdots
    p_k \big( 2c + \mathbb{E}[A_{k-1}] \big) + p_{pred} \big( 2c + \mathbb{E}[A_{k-1}] \big)
    \end{equation*}
    \begin{equation*}
    \leq p_1 ( 2c + k\cdot c ) +
    \cdots + p_{i^*}\cdot c + \cdots
    p_k ( 2c + k\cdot c ) + p_{pred} ( 2c + k \cdot c )
    \end{equation*}
    \begin{equation*}
        =\bigg[ (k+2) \bigg( \sum_{\substack{i=1 \\ i\neq i^*}}^{k} p_i + p_{pred} \bigg) + p_{i^*} \bigg] c
    \end{equation*}
    From (\ref{eq_one}) we get:
    \begin{equation*}
        \mathbb{E}[ALG] \leq \big[ (k+2) (1-p_{i^*}) + p_{i^*} \big]  c
    \end{equation*}
    \begin{equation*}
        = \big[ k -(k+1) p_{i^*}+2 \big] c
    \end{equation*}
    From (\ref{probabilities}) it follows that:
    \begin{equation*}
        \mathbb{E}[ALG] \leq \bigg[ k-(k+1) \frac{\epsilon}{k(k+1)}+2 \bigg] c
    \end{equation*}
    \begin{equation*}
        \leq \bigg[ k+\frac{2k-\epsilon}{k} \bigg] c
    \end{equation*}
    We have that:
    \begin{equation*}
        \frac{2k-\epsilon}{k} \leq \frac{k}{\epsilon}
    \end{equation*}
    \begin{equation*}
        \Leftrightarrow 2\epsilon \cdot k-\epsilon^2 \leq k^2
    \end{equation*}
    \begin{equation*}
        \Leftrightarrow (k-\epsilon)^2 \geq 0 \text{, which is always true.}
    \end{equation*}
    Therefore, we finally get:
    \begin{equation*}
        \mathbb{E}[ALG] \leq \bigg( k+\frac{k}{\epsilon} \bigg) OPT
    \end{equation*}
\end{itemize}

\end{proof}

\section{Missing proofs from Section~\ref{sec:robustness}}\label{app:b}

\subsection{Proof of Theorem~\ref{theorem1}}\label{app:2k+1}

{\bf Theorem~\ref
{theorem1}.} {\it 
There is no deterministic online algorithm with predictions that
achieves competitive ratio smaller than $2k+1$, even when the prediction has error at most $2$ and the graph is path-disjoint.}

\begin{proof} 

Consider a graph $\mathcal{G}^*$ with $k+1$ \iffalse source-sink \fi paths $P_1$, $P_2$,..., $P_{k}$, $P_{k+1}$, which are node-disjoint. All the paths have equal costs, meaning that $c_1=c_2=...=c_k=c_{k+1}$. $P_1$, $P_2$,..., $P_{k}$ are predicted to be blocked ($k$ predicted blocks).

Every deterministic algorithm corresponds to a permutation,
which describes in which order the paths are being explored. The adversary can only create an instance with error no more than $2$.

The adversary blocks every path tried by a deterministic algorithm except for the last one ($k$ blocks). Then the adversary creates an instance with a prediction error of at most $2$ and the algorithm has competitive ratio at least $2k+1$.
\end{proof}

\subsection{Proof of Theorem~\ref{theorem2k-1}.}\label{app:blabla}

{\bf Theorem~\ref{theorem2k-1}}{\it 
There is a $(2k-1)$-competitive algorithm when the prediction error is at most $1$ and $k\geq 3$ is known. 
}
\begin{proof}
The claimed algorithm, called \textsc{Err1-Backtrack}, is described in Algorithm~\ref{algo_err1BT} in a recursive form. %The two main ingredients are (1) a careful comparison of the lengths of the shortest path and the shortest path without predicted blocked edges, to decide which one to explore, and (2) the fact that when the error is at most 1, if a blocked edge is discovered and was not predicted to be so, then we know exactly the set of blocked edges (i.e., the predicted ones and the new one) and thus we can determine directly the optimal solution without further testing. 

We prove the result by induction on $k$. The inductive step is easy: suppose that \textsc{Err1-Backtrack} is $(2k-3)$-competitive for $k-1\geq 3$, and let us consider an instance with parameters $k$ and $B$. If $P_1$ is not blocked, the algorithm is optimal. If $e\not\in B$ is found, then since the error is at most 1 we know that $B\cup\{e\}$ is exactly the set of blocked edges. Then the cost of the algorithm is at most $2c_1+OPT\leq 3OPT$. Finally, by induction, if the recursive step (line~\ref{algo:test199}) is called the cost of the algorithm is at most $2c_1+(2k-3)OPT\leq (2k-1) OPT$.\\

We now focus on the case $k=3$. We have to show that the algorithm is 5-competitive.

Let us first consider the case where condition $c_a\leq \alpha_ic_i$ in line~\ref{alg:test1} occurs. 
\begin{itemize}
    \item If it holds at $i=1$, then $c_a\leq \alpha_1c_1=5c_1$. When applying \textsc{Backtrack}, either the first explored path is not blocked, or a new blocked edge $e\not\in B$ is found. In the first case, the cost of the algorithm is $c_a\leq 5 c_1\leq 5 OPT$. In the latter case, the blocked edges are {\it exactly} $B\cup \{e\}$, so \textsc{Backtrack} is 3-competitive.
    \item  If it holds at $i=2$, then $c_a\leq \alpha_2c_2=\frac{25c_2}{7}$, and $c_a> \alpha_1c_1=5c_1$. Similarly, either the first explored path is not blocked, or a new blocked edge $e\not\in B$ is found. In the first case, the cost of the algorithm is $2c_1+c_a\leq (2/5+1) c_a=7c_a/5\leq 5 c_2\leq 5 OPT$. In the latter case, the blocked edges are exactly $B\cup \{e\}$, so \textsc{Backtrack} is 3-competitive, and the cost is at most $2c_1+3OPT\leq 5 OPT$. 
    \end{itemize}
    
    If in line~\ref{alg:test2} $P_i$ is not blocked then the algorithm is optimal if $i=1$, and the cost of the algorithm is $2c_1+c_2\leq 3 OPT$ if $i=2$.
    
    Let us now focus on the condition line~\ref{alg:test3}. Here again, if $e_i\not \in B$ is blocked, then the set of blocked edges is $B\cup\{e_i\}$, so we are able to directly find the optimal solution. Then the cost of the algorithm is at most (using $c_a\leq OPT$ here) $2c_1+2c_2+OPT\leq \left(\frac{2}{5}+\frac{14}{25}+1\right) OPT\leq 5 OPT$.

    The last case to check corresponds to line~\ref{alg:test4}. At this stage, we know that $c_a>\alpha_i c_i, i=1,2$, and that $e_1\in B$ and $e_2\in B$ are blocked. Then, as $k= 3$, there is at most 1 remaining blocked edge, so \textsc{Backtrack} is $3$-competitive. 
    \begin{itemize}
        \item If $B=\{e_1,e_2\}$, then $c_a\leq OPT$ and the cost of the algorithm is  $2c_1+2c_2+3OPT\leq \left(\frac{2}{5}+\frac{14}{25}+3\right) OPT\leq 5 OPT$.
        \item Otherwise, $B=\{e_1,e_2,e_3\}$. If $e_3$ is not blocked then the set of blocked edges is $\{e_1,e_2\}$, so \textsc{Backtrack} directly finds the optimum, and the cost is $2c_1+2c_2+OPT\leq 5 OPT$. Otherwise, $e_3$ is blocked, so $c_a\leq OPT$ and we get again that the cost is at most $2c_1+2c_2+3OPT\leq \left(\frac{2}{5}+\frac{14}{25}+3\right) OPT\leq 5 OPT$.
    \end{itemize} 
\end{proof}

\IncMargin{1em}
\begin{algorithm}
\SetKwData{Left}{left}\SetKwData{This}{this}\SetKwData{Up}{up}
\SetKwFunction{Union}{Union}\SetKwFunction{FindCompress}{FindCompress}
\SetKwInOut{Input}{Input}\SetKwInOut{Output}{Output}
\Input{An instance of CTP with prediction with parameter $k$, $B$ a set of predicted blocked edges}
\Output{An $s-t$ path}
\BlankLine

\uIf{$k>3$}{
    $P_1,c_1\leftarrow$ a shortest path, and its cost;
    
    Explore $P_1$;
    
    \lIf{$P_1$ is not blocked}{
    Return $P_1$
    }
    \lElseIf{a blocked edge $e\not\in B$ is
      found}{
      Remove $B\cup\{e\}$ and output a shortest path 
    }
    \lElse{Apply \textsc{Err1-Backtrack} with parameters $k-1$ and $B\setminus e$, where $e$ is the blocked edge}\label{algo:test199}
}    
\Else{
    $P_a,c_a\leftarrow$ a shortest path when $B$ is removed, and its cost; 
   
   $\alpha_1\leftarrow 5$, $\alpha_2\leftarrow \frac{25}{7}$; 
   
    \For{$i=1,2$}{
      $P_i,c_i\leftarrow$ a shortest path, and its cost;
   
      \lIf{$c_a\leq \alpha_i c_i$}{
         Apply \textsc{Backtrack} on the graph without $B$ and output the found path
      }\label{alg:test1}
      \Else{
        Explore $P_i$;
        
        \lIf{$P_i$ is not blocked}{Return $P_i$}\label{alg:test2}
        \lElseIf{a blocked edge $e_i\not\in B$   is found}{Remove $B\cup\{e_i\}$ and    output a shortest path
        }\label{alg:test3}
        \lElse{Remove the blocked edge $e_i\in B$ (and continue)
        }
      }
    }
 Apply \textsc{Backtrack} and output the found shortest path.\label{alg:test4}
}

\caption{\textsc{Err1-Backtrack}}\label{algo_err1BT}
\end{algorithm}\DecMargin{1em}

\subsection{Proof of Theorem~\ref{theorem2}}\label{app:theorem2}

{\bf Theorem~\ref{theorem2}.}
{\it There is no deterministic  algorithm that 
achieves competitive ratio smaller than $2k-1$, even when the prediction has $error$ at most $1$ and the graph is path-disjoint.}

\begin{proof}
Consider a graph $\mathcal{G}^*$ with $k+1$ \iffalse source-sink \fi paths $P_1$, $P_2$,..., $P_{k}$, $P_{k+1}$, which are node-disjoint.  $P_1$ is predicted to be unblocked and all the other paths are predicted to be blocked ($k$ blocks predicted). The paths $P_2$, $P_3$, $P_4$, ..., $P_{k+1}$ have costs $c_2=c_3=...=c_{k+1}=1$ and path $P_1$ has cost $c_1 > (2k+1)$.
    
Every deterministic algorithm corresponds to a permutation, which describes in which order the paths are being explored. The adversary can only create an instance with error no more than $1$.

\begin{itemize}

    \item If a deterministic algorithm chooses path $P_1$ before traversing all $k$ paths $P_2$, $P_3$, ..., $P_{k+1}$, then the adversary chooses a path $P_j$ ($2\leq j \leq k+1$), which has not been traversed yet, to be unblocked.
    
    In that case the algorithm has competitive ratio $r \geq \frac{c_1}{c_j}=c_1>2k+1$. 
    
    \item If a deterministic algorithm traverses first all $k$ paths $P_2$, $P_3$, ..., $P_{k+1}$, then the adversary chooses all those paths to be blocked except the last one which is the optimal. This incurs a cost of $(2(k-1) + 1)\cdot c_j= (2k-1)\cdot c_j$.
    
    Then the algorithm has a competitive ratio of $r \geq 2k-1$. 
    
\end{itemize}

Consequently, there is no deterministic prediction-augmented online algorithm with competitive ratio smaller than $2k-1$, even when the prediction has an error of $1$.
\end{proof}

\subsection{Proof of Theorem~\ref{theor_k_1}}\label{app:theor_k_1}

{\bf Theorem~\ref{theor_k_1}} {\it 
When $k=1$, there is no deterministic  algorithm   that achieves competitive ratio smaller than $3$, even when the prediction has error at most $1$ and the graph is path-disjoint.}

\begin{proof}
Consider a graph $\mathcal{G}^*$ with $2$ \iffalse source-sink \fi paths $P_1$, $P_2$, which are node-disjoint. All the paths have equal costs, meaning that $c_1=c_2$. All paths are predicted to be unblocked. 

Every deterministic algorithm corresponds to a permutation,
which describes in which order the paths are being explored. The adversary can only create an instance with error no more than $1$.

The adversary blocks the first path tried by a deterministic algorithm. Then it is clear that the algorithm has a competitive ratio of at least $3$.

Consequently, there is no deterministic prediction-augmented online algorithm with competitive ratio smaller than $3$, even when $error = 1$ ($k=1$).
\end{proof}

\subsection{Proof of Theorem~\ref{theoremLBk=2}}\label{app:LBk=2}

{\bf Theorem~\ref{theoremLBk=2}.} {\it
When $k=2$, there is no deterministic algorithm that achieves competitive ratio smaller than $\frac{3+\sqrt{17}}{2}$, even when the prediction has error at most $1$ and the graph is path-disjoint.}

\begin{proof}
Consider a graph $\mathcal{G}^*$ with $3$ \iffalse source-sink \fi paths $P_1$, $P_2$, $P_3$ which are node-disjoint and their corresponding costs are $c_1=1$, $c_2=c_3=\frac{3+\sqrt{17}}{2}$. Only $P_1$ is predicted to be blocked. 

Every deterministic algorithm corresponds to a permutation,
which describes in which order the paths are being explored. The adversary can only create an instance with error no more than $1$.

If a deterministic algorithm first chooses to traverse path $P_1$, then the adversary blocks that path and the next one tried ($P_2$ or $P_3$). This instance has an error of $1$ and any such algorithm has competitive ratio $\frac{2c_1+2c_2+c_3}{c_3}$ or $\frac{2c_1+2c_3+c_2}{c_2}$. Therefore, it achieves a competitive ratio of 
\begin{equation*}
    r=\frac{2c_1+2c_2+c_3}{c_3}= \frac{2c_1+2c_3+c_2}{c_2}=\frac{3+\sqrt{17}}{2}
\end{equation*}

On the other hand, if a deterministic algorithm chooses at first to traverse path $P_2$ or $P_3$, then the adversary unblocks $P_1$ (doesn't block $P_2$ or $P_3$ and creates an instance with $error=1$). Hence, all paths are feasible and the algorithm has competitive ratio
\begin{equation*}
    r=\frac{c_2}{c_1}=\frac{c_3}{c_1}=\frac{3+\sqrt{17}}{2}
\end{equation*}

Consequently, there is no deterministic prediction-augmented online algorithm with competitive ratio smaller than $\frac{3+\sqrt{17}}{2}$, even when the prediction has error at most $1$ ($k=2$).
\end{proof}

\subsection{Proof of Theorem~\ref{theoremUBk=2}}\label{app:theoremUBk=2}

{\bf Theorem~\ref{theoremUBk=2}}
{\it There exists a $\frac{3+\sqrt{17}}{2}$-competitive algorithm when the prediction error is at most $1$ and $k=2$.}

\begin{proof}
Algorithm \textsc{Err1-Backtrack2}  is described in Algorithm~\ref{algo_err1BT2} in a recursive form. The algorithm is a simple modification of \textsc{Err1-Backtrack} when $k \leq 3$ and the analysis is similar.
 
\IncMargin{1em}
\begin{algorithm}
\SetKwData{Left}{left}\SetKwData{This}{this}\SetKwData{Up}{up}
\SetKwFunction{Union}{Union}\SetKwFunction{FindCompress}{FindCompress}
\SetKwInOut{Input}{Input}\SetKwInOut{Output}{Output}
\Input{An instance of CTP with prediction with parameter $k=2$, $B$ a set of predicted blocked edges}
\Output{An $s-t$ path}
\BlankLine
    $P_a,c_a\leftarrow$ a shortest path when $B$ is removed, and its cost; 
   
   $\alpha_1\leftarrow \frac{3+\sqrt{17}}{2}$, $\alpha_2\leftarrow 1$; 
   
    \For{$i=1,2$}{
      $P_i,c_i\leftarrow$ a shortest path, and its cost;
   
      \lIf{$c_a\leq \alpha_i c_i$}{
         Apply \textsc{Backtrack} on the graph without $B$ and output the found path
      }\label{alg2:test1}
      \Else{
        Explore $P_i$;
        
        \lIf{$P_i$ is not blocked}{Return $P_i$}\label{alg2:test2}
        \lElseIf{a blocked edge $e_i\not\in B$   is found}{Remove $B\cup\{e_i\}$ and    output a shortest path
        }\label{alg2:test3}
        \lElse{Remove the blocked edge $e_i\in B$ (and continue)
        }
      }
    }
 Apply \textsc{Backtrack} and output the found shortest path.\label{alg2:test4}

\caption{\textsc{Err1-Backtrack2}}\label{algo_err1BT2}
\end{algorithm}\DecMargin{1em}

Let us first consider the case where condition $c_a\leq \alpha_ic_i$ in line~\ref{alg2:test1} occurs. 
\begin{itemize}
    \item If it holds at $i=1$, then $c_a\leq \alpha_1c_1=\frac{3+\sqrt{17}}{2} \cdot  c_1$. When applying \textsc{Backtrack}, either the first explored path is not blocked, or a new blocked edge $e\not\in B$ is found. In the first case, the cost of the algorithm is $c_a\leq \frac{3+\sqrt{17}}{2} \cdot c_1\leq \frac{3+\sqrt{17}}{2} \cdot OPT$. In the latter case, the blocked edges are {\it exactly} $B\cup \{e\}$, so \textsc{Backtrack} is $3$-competitive.
    \item  If it holds at $i=2$, then $c_a\leq \alpha_2c_2=c_2$, and $c_a> \alpha_1c_1=\frac{3+\sqrt{17}}{2} \cdot c_1$. Similarly, either the first explored path is not blocked, or a new blocked edge $e\not\in B$ is found. In the first case, the cost of the algorithm is $2c_1+c_a\leq \bigg(\frac{4}{3+\sqrt{17}}+1\bigg) c_a \leq \bigg(\frac{4}{3+\sqrt{17}}+1\bigg) c_2 < 2c_2 \leq 2 OPT$. In the latter case, the blocked edges are exactly $B\cup \{e\}$, so \textsc{Backtrack} is $3$-competitive, and the cost is at most $2c_1+3 OPT\leq \bigg(\frac{4}{3+\sqrt{17}}+3\bigg) OPT = \frac{3+\sqrt{17}}{2} \cdot OPT$. 
    \end{itemize}
    
    If in line~\ref{alg2:test2} $P_i$ is not blocked then the algorithm is optimal if $i=1$, and the cost of the algorithm is $2c_1+c_2\leq 3 OPT$ if $i=2$.
    
    Let us now focus on the condition line~\ref{alg2:test3}. Here again, if $e_i\not \in B$ is blocked, then the set of blocked edges is $B\cup\{e_i\}$, so we are able to directly find the optimal solution. Then the cost of the algorithm is at most (using $c_a\leq OPT$ here) $2c_1+2c_2+OPT \leq \left(\frac{4}{3+\sqrt{17}}+2+1\right) OPT = \frac{3+\sqrt{17}}{2} \cdot OPT$.

    The last case to check corresponds to line~\ref{alg2:test4}. At this stage, we know that $c_a>\alpha_i c_i, i=1,2$, and that $e_1\in B$ and $e_2\in B$ are blocked. Then, as $k= 2$, there are no more blocked edges, so \textsc{Backtrack} finds directly the optimal path. Then $c_a \leq OPT$ and the cost of the algorithm is
    $2c_1+2c_2+OPT \leq \left(\frac{4}{3+\sqrt{17}}+2+1\right) OPT = \frac{3+\sqrt{17}}{2} \cdot OPT$.
     
\end{proof}

%\subsection{Proof of Theorem~\ref{theor_rand_k1}}\label{app:nokeeps}

%{\bf Theorem~\ref{theor_rand_k1}}
%{\it 
%There is no randomized algorithm  that 
%achieves competitive ratio smaller than $k+1$ against an oblivious adversary, even when the prediction has error at most $2$ and the graph is path-disjoint.} 
%The parameter $k$, here, bounds the re\refal number of blocked edges fromepsab{\bf Td{theo~\ref rem{theor_rand_k1}

\subsection{Proof of Theorem~\ref{th:randk}}\label{app:noke1}

{\bf Theorem~\ref{th:randk}.}
{\it There is no randomized algorithm  that 
achieves competitive ratio smaller than $k$ against an oblivious adversary, even when the prediction has error at most $1$ and the graph is path-disjoint.}

\begin{proof}
Consider a graph $\mathcal{G}^*$ with $k+1$ \iffalse source-sink \fi paths $P_1$, $P_2$, ..., $P_k$, $P_{k+1}$, which are node-disjoint. The paths $P_1$, $P_2$,..., $P_{k-1}$, $P_{k}$ have costs equal to $1$ ($c_1 = c_2 = ... = c_k =1$) and path $P_{k+1}$ has cost $c_{k+1} \geq k+1$.

The goal here is to prove the theorem by applying Yao's Principle. We choose $i \in \{1,...,k\}$ uniformly at random and block all paths $P_j$ with $j \in \{1,...,k\}$ and $j \ne i$ ($k-1$ blocked edges as all paths are node-disjoint). Hence, we have two feasible paths $P_i$ and $P_{k+1}$. All paths except for $P_{k+1}$ are predicted to be blocked. This prediction has an error of $1$.

The optimal offline cost is $1$ (path $P_i$). If a deterministic algorithm chooses the other feasible path $P_{k+1}$, then it has a competitive ratio of at least $k+1$. So, any deterministic online algorithm can't achieve a competitive ratio smaller than $k+1$ when choosing at any step path $P_{k+1}$. 

This is possible only if it finds path $P_i$. An arbitrary deterministic online algorithm finds path $P_i$ on the $l$th trial for $l=1,...,k$ with probability $\frac{1}{k}$ as all the first $k$ paths are predicted to be blocked.

If the algorithm is successful on its $l$th try, it incurs a cost of $2l-1$, and thus it has an expected cost of at least
\begin{equation*}
    cost \geq \frac{1}{k} \sum_{l=1}^{k}(2l-1) = \frac{1}{k} \cdot k^{2} = k
\end{equation*}

Therefore, the expected cost of any deterministic algorithm is at least $k$ and the optimal cost is $1$. It follows from Yao's Principle that the competitive ratio of any prediction-augmented randomized online algorithm is at least $k$, even when the prediction has an error of $1$.
\end{proof}

\end{document}